\documentclass[journal]{IEEEtran}

\IEEEoverridecommandlockouts
\usepackage{amsthm}
\usepackage{wrapfig}
\usepackage{amsthm}
\usepackage{mathrsfs}
\usepackage{graphicx}
\usepackage{graphics}
\usepackage{amssymb}
\usepackage{amsmath,mathtools}
\usepackage{color}
\usepackage{xspace}
\usepackage{algpseudocode}
\usepackage{bbm}
\usepackage{comment}
\usepackage{cite}
\usepackage{algorithmicx}
\usepackage{algorithm}
\usepackage{psfrag}
\usepackage{url}
\usepackage{float}
\usepackage{subfig}
\usepackage{soul}

\usepackage{rotating}
\usepackage{chngcntr}
\usepackage{apptools}
\usepackage{graphicx}
\usepackage{graphics}
\usepackage{amssymb}
\usepackage{mathtools}
\usepackage{color}
\usepackage{xspace}
\usepackage{algpseudocode}
\usepackage{bbm}
\usepackage{comment}
\usepackage{algorithmicx}
\usepackage{psfrag}

\usepackage{amsmath}
\usepackage{upgreek}
\allowdisplaybreaks[4]
\usepackage{tikz}
\usetikzlibrary{shapes,arrows}
\usetikzlibrary{positioning}

\tikzstyle{block}=[draw opacity=0.7,line width=1.4cm]
\DeclareMathAlphabet{\mathpzc}{OT1}{pzc}{m}{it}
\definecolor{CranJ}{cmyk}{0,0.69,0.54,0.04} 
\definecolor{PinkJ}{cmyk}{0,0.71,0.43,0.12} 
\definecolor{Cran}{cmyk}{0,0.73,0.41,0.29} 
\definecolor{VRed}{cmyk}{0,0.75,0.25,0.2} 
\definecolor{ORed}{cmyk}{0,0.75,0.75,0} 
\definecolor{CBlue}{cmyk}{1,0.25,0,0} 

\newcommand{\KLD}{\operatorname{D}_{\text{KL}}}
\newcommand{\tr}{\operatorname{tr}}
\newcommand{\Con}{\operatorname{Con}}
\newcommand{\ee}{\operatorname{e}}

\newcommand{\VV}{\mathcal{V}}
\newcommand{\EE}{\mathcal{E}}
\newcommand{\GG}{\mathcal{G}}

\newcommand{\real}{{\mathbb{R}}}
\newcommand{\reals}{{\mathbb{R}}}

\newcommand{\realpositive}{{\mathbb{R}}_{>0}}
\newcommand{\realnonnegative}{{\mathbb{R}}_{\ge 0}}

\newcommand{\vect}[1]{\boldsymbol{\mathbf{#1}}}

\newcommand{\dvect}[1]{\dot{\vect{#1}}}

\newcommand{\boxend}{\hfill \ensuremath{\Box}}

\newtheorem{thm}{Theorem}

\newtheorem{rem}[thm]{Remark}

\newcommand{\oprocendsymbol}{\hbox{$\bullet$}}
\newcommand{\oprocend}{\relax\ifmmode\else\unskip\hfill\fi\oprocendsymbol}

\title{\LARGE \bf A distributed service-matching coverage via\\ heterogeneous mobile agents 
} 

\author{Yi-Fan Chung and Solmaz S. Kia, \emph{Senior Member, IEEE} 
\thanks{The authors are with the Department of Mechanical and Aerospace Engineering, University of California, Irvine, Irvine, CA 92697, {\tt \{yfchung,solmaz\}@uci.edu}. This work is supported by NSF award IIS-SAS-1724331.}}

\begin{document}
\maketitle
\begin{abstract}
We propose a distributed deployment solution for a group of mobile agents that should provide a service for a dense set of targets. The agents are heterogeneous in a sense that their quality of service (QoS), modeled as a spatial Gaussian distribution, is different. To provide the best service, the objective is to deploy the agents such that their collective QoS distribution is as close as possible to the density distribution of the targets. We propose a distributed consensus-based expectation-maximization (EM) algorithm to estimate the target density distribution, modeled as a Gaussian mixture model (GMM). The GMM not only gives an estimate of the targets’ distribution, but also partitions the area to subregions, each of which is represented by one of the GMM’s Gaussian bases. We use the Kullback-Leibler divergence (KLD) to evaluate the similarity between the QoS distribution of each agent and each Gaussian basis/subregion. Then, a distributed assignment problem is formulated and solved as a discrete optimal mass transport problem that allocates each agent to a subregion by taking the KLD as the assignment cost. We demonstrate our results by a sensor deployment for event detection where the sensor's QoS is modeled as an anisotropic Gaussian distribution.

\end{abstract}



Over the last decade, deploying a group of networked mobile agents to cover a region with a service objective such monitoring, data collection, and wireless communication have attracted considerable attention,  
see for a few examples\cite{KM-Ck-HB:11,FA-MAN-MHJ-GMC-TAJ-KSB:14,SS-BF-MS:08}. The deployment strategy commonly includes partitioning the environment and allocating agents to those partitions. That is, the area of interest is partitioned into subregions and each agent is allocated to a location in the subregion such that some coverage metric is optimized. 

The classic Voronoi-based deployment strategy~\cite{CJ-MS-KT-BF:04,BF-CJ-MS:09,WG-CG-LP-TF:06,PL-KV-MR-PG:08,PA-FLC-PL-SM:15,AO-KDE:16,LK-KJ:09,FF-ZX-CX-ZT:17,GA-HS-HT-FM:08,SM-RD-SJ:09,CA-TM-CR-SL-PG:15,SF-MM-ZY-GBW:15} is a prime example of multi-agent deployment for area coverage.~\cite{CJ-MS-KT-BF:04} as one of the initial work in this area develops a deployment algorithm based on the Llyod method to compute the Voronoi partition and allocate the agents to the Centroidal Voronoi configuration which is well-known as the optimal configuration of a class of locational optimization cost function~\cite{BF-CJ-MS:09}. 
The original Voronoi-based deployment strategy is developed for homogeneous agents. To reach the optimal coverage with heterogeneous agents whose service capabilities are different,~\cite{PL-KV-MR-PG:08,PA-FLC-PL-SM:15,AO-KDE:16} employ the weighted Voronoi diagram where the weightings account for heterogeneity among the agents. 
The works mention above assume the footprint of the service provided by an agent is disk-shaped, i.e., the distribution of QoS is isotropic. However, an anisotropic service model is more realistic because sensory systems such as cameras, directional antenna, and radars are anisotropic. \cite{LK-KJ:09,FF-ZX-CX-ZT:17} for wedge-shape and \cite{GA-HS-HT-FM:08} for elliptic footprint adapt an anisotropic service model by the modifying Voronoi diagrams to match the features of the anisotropy of the sensors. But these methods increase the complexity of the Voronoi partition, which make the design of distributed optimal deployment strategies very challenging.
The heterogeneity in deployment algorithm design can also be due to non-uniformity in area of interest.  To deal with such scenarios, a priority (sensory) function of the position is introduced to indicate the importance level over the area, where a location needs higher QoS if the value of the priority function is higher at that location. 
The work~\cite{CJ-MS-KT-BF:04,BF-CJ-MS:09,WG-CG-LP-TF:06,PL-KV-MR-PG:08,PA-FLC-PL-SM:15,AO-KDE:16,LK-KJ:09,FF-ZX-CX-ZT:17,GA-HS-HT-FM:08} mentioned  above assume the priority function is known to each agent. This assumption may not be realistic for every application. 
\cite{SM-RD-SJ:09} uses the parameterized basis functions to model the priority distribution, and \cite{CA-TM-CR-SL-PG:15} models the distribution by a zero-mean Gaussian random field. Then, in both~\cite{SM-RD-SJ:09} and~\cite{CA-TM-CR-SL-PG:15}, the agents gradually fit their model to the true distribution using their local sensor measurement while exploring the area. In~\cite{SF-MM-ZY-GBW:15}, the authors assume the unknown priority function is a function of the position of some unknown targets. The search agents aim to detect the targets while exploring the area, and then, broadcast their information about the environment to the service agents so the service agents can focus on the deployment problem.

In this paper, we propose a novel deployment strategy for mobile agents to cover a collection of dense targets with their heterogeneous anisotropic services. We model the  agents' QoS distributions by Gaussian distributions with difference covariance matrix reflecting the difference in the capabilities of the agents. Since the footprint of Gaussian distribution is elliptic, the agents' QoS are heterogeneous and anisotropic. Furthermore, the density distribution of the targets is considered as the priority function which is unknown a priori. Our deployment objective is deploying the agents such that the resulting QoS distribution of agents is similar to the density distribution of the targets. Hence, the agents' service efficiently covers the targets; that is, the place containing more (resp. fewer) targets is served with higher (resp. lower) QoS.

We model the unknown density distribution of the targets by a GMM. We propose a distributed consensus-based EM algorithm to enable the agents to learn the parameters of the GMM. With the proposed distributed EM, we do not require each agent to measure the targets locally. Only a subset of the agents can measure the targets and their information is then propagated through the consensus protocol to all agents.
Moreover, the GMM intrinsically partitions the area into a set of subregions, each of which represents a Gaussian basis. Therefore, after estimating the target density distribution, the agents also complete the area partitioning task. We note that unlike the distributed Voronoi partition that requires the agents to be able to communicate to their Voronoi neighbors, which at times can be unrealistic because the physical distance between Voronoi neighbors may but them outside of the communication range of each other, our approach only requires the communication graph among the agents to be connected. We use the KLD measure to assess the similarity of the QoS provided by the agents and the targets' density distributions. We propose to obtain the optimal deployment pose (position and orientation) of the service agents by minimizing this KLD measure. Since this KLD measure is highly coupled and computing a distributed solution for it is challenging, we propose a suboptimal deployment solution in the form of an optimal mass transport problem to allocate each agent to a Gaussian basis subregion of the GMM used to estimate the targets' distribution, where the cost of transporting the agent to a subregion is the KLD value between the agent's QoS distribution and the Gaussian basis distribution. We show that this assignment problem is a distributed linear programming that can be solved efficiently by the distributed simplex algorithm of~\cite{BM-NG-BF-AF:12}.
We illustrate our results via an application in the deployments of sensor network for event~detection.

\section{Notations and Preliminaries}\label{sec::notaion_prelim_relay}
\vspace{-0.08in} 
We let $\reals$, $\realpositive$, $\realnonnegative$, $\mathbb{Z}$, $\mathbb{Z}_{> 0}$ and $\mathbb{Z}_{\geq 0}$
denote the set of real, positive real, non-negative real, integer, positive integer, and non-negative integer, respectively. For $\vect{s}\in\reals^d$,
$\|\vect{s}\|=\sqrt{\vect{s}^\top\vect{s}}$ denotes the standard
Euclidean norm. We let
$\vect{1}_n$ (resp. $\vect{0}_{n}$) denote the vector of $n$ ones
(resp. $n$ zeros), and $\vect{I}_n$ denote the $n\times n$ identity
matrix.  Given two continuous probability density distributions $p(\vect{x})$ and $q(\vect{x})$, $\vect{x}\in\mathbb{X}$, the \emph{Kullback–Leibler divergence} (KLD) is defined as $\KLD\big(p(\vect{x})||q(\vect{x})\big)=\int_{\vect{x}\in\mathbb{X}}p(\vect{x})\ln \frac{p(\vect{x})}{q(\vect{x})}d\vect{x}.$
KLD is a  measure of similarity (dissimilarity) between two probability distributions $p(\vect{x})$ and $q(\vect{x})$, where the smaller the value the more similar two distributions are. KLD is zero if and only if the two distribution are identical \cite[p.34]{MDJC:03}.  For Gaussian distributions, $p(\vect{x})=\mathcal{N}(\vect{\mu}_0,\vect{\Sigma}_0)$ and  $q(\vect{x})=\mathcal{N}(\vect{\mu}_1,\vect{\Sigma}_1)$, the KLD has a closed form expression~\cite[eq. (2)]{HJR-ORA:07}
\begin{align}\label{eq::close_KLD}
\KLD \Big(p(\vect{x})||q(\vect{x})\big)=& \frac{1}{2}\big(\ln\frac{|\vect{\Sigma}_1|}{|\vect{\Sigma}_0|}+(\vect{\mu}_0-\vect{\mu}_1)^\top\vect{\Sigma}_1^{-1}(\vect{\mu}_0-\vect{\mu}_1)\nonumber\\&\quad+\tr(\vect{\Sigma}_1^{-1}\vect{\Sigma}_0)-n \Big),
\end{align}
where $n$ is the dimension of the distributions.

We follow~\cite{FB-JC-SM:09} for Our graph theoretic notations and definitions. 
A \emph{graph}, is a triplet $\GG = (\VV ,\EE,
\vect{\sf{A}})$,~where $\VV=\{1,\dots,N\}$ is the \emph{node set} and
$\EE \subseteq \VV\times \VV$ is the \emph{edge set}, and $\vect{\sf{A}}\in\real^{N\times N}$ is a \emph{adjacency}
matrix such that $ \sf{a}_{ij} =1$ if $(i, j) \in\EE$ and $
\mathsf{a}_{ij} = 0$, otherwise.  
An edge $(i, j)$ from $i$ to $j$ means that agents $i$ and $j$ can communicate. 
A \emph{path} is a sequence of nodes
connected by edges. A \emph{connected graph} is an undirected graph in which for
every pair of nodes there is a path connecting them. 

To develop our distributed density estimator in Section~\ref{sec::stage1}, we rely on the \emph{dynamics active weighted average consensus algorithm} that is shown in Algorithm~\ref{alg::consensus}. In a dynamics active weighted average consensus at any time, only a subset of the agents are active, meaning that only a subset of agents collects measurements $\mathsf{r}^i$. The objective then is to enable all the agents, both active and passive, to obtain the weighted average of the collected measurements,  $\frac{\sum_i\eta^i(l) \mathsf{r}^i(l)}{\sum_i \eta^i(l)}$, without knowing the set of active agents. Here, $ \eta^i(l)=0$ if $i$ is passive at time step $l$ and $ \eta^i(l)\in\real_{>0}$ if $i$ is active. 
\cite{YC-SS:20} shows that Algorithm~\ref{alg::consensus}, starting at any $z^i,y^i(0)\in\real$, makes $\vect{y}^i(l)$ track the time varying weighted average signal $\frac{\sum_i\eta^i(l) \mathsf{r}^i(l)}{\sum_i \eta^i(l)}$ with a bounded tracking error as $l\to\infty$. Moreover, if the weights and reference signals are static, the tracking error vanishes with time, i.e., $\lim_{l\to\infty}\vect{y}^i(l)=\frac{\sum_i\eta^i \mathsf{r}^i}{\sum_i \eta^i}$.


\begin{algorithm}[t]
{\small
\caption{Active weighted average consensus algorithm~\cite{YC-SS:20} }\label{alg::consensus}
\footnotesize
$[\vect{y}^i,\vect{z}^i,\vect{v}^i]\gets \Con(\eta^i,\vect{\mathsf{r}}^i,\vect{z}^i_0,\vect{v}^i_0)$\\
\begin{algorithmic}
\Require Weight $\eta^i$, reference $\vect{\mathsf{r}}^i$, number of loops $L$, a small enough number $\delta_c>0$. \\
\textbf{Initialization:} $\vect{z}^i(1)=\vect{z}^i_0$ and $\vect{v}^i(1)=\vect{v}^i_0$

\For{$l=1:L$}
\begin{align*}
 \vect{y}^i(l)&=\vect{z}^i(l)+\eta^i(l) \vect{\mathsf{r}}^i(l),\\
 \vect{z}^i(l+1)&=\vect{z}^i(l)-\delta_c\eta^i(l)(\vect{y}^i(l)-\vect{\mathsf{r}}^i(l))\\
 &\!\!\!\!\!\!\!\!\!\!\!\!\!\!\!\!\!\!\!\!-\delta_c\sum_{j=1}^N\mathsf{a}_{ij}(\vect{y}^i(l)-\vect{y}^j(l))-\delta_c \sum_{j=1}^N\mathsf{a}_{ij}(\vect{v}^i(l)-\vect{v}^j(l)),\nonumber\\
 \vect{v}^i(l+1)&=\vect{v}^i(l)+\delta_c\sum\nolimits_{j=1}^N\mathsf{a}_{ij}(\vect{y}^i(l)-\vect{y}^j(l)).
\end{align*}
\EndFor\\
\Return $\vect{y}^i(l),\vect{z}^i(l+1),\vect{v}^i(l+1)$
\end{algorithmic}
}
\end{algorithm}


\section{Problem Definition and Objective}\label{sec::prob_def}
\vspace{-0.08in} 
We consider a deployment problem for a group of mobile agents over a set of dense targets $\{\vect{x}_t^n\}_{n=1}^M\subset \real^2$ on a planar ground with the objectives such as event detection,  wireless communication or monitoring, which we refer to it in general term as providing a `service'. The probability density distribution $p(\vect{x})$, $\vect{x}\in\real^2$, which is unknown to the agents, represents the density distribution of the targets. The mobile agents communicating over a connected undirected graph $\GG = (\VV ,\EE,
\vect{\sf{A}})$ consist of two types. There are a set $\VV_{\textup{a}}\subseteq\VV$ of \emph{active} agents that have the capability to actively detect the targets and a set $\VV_{\textup{s}}=\{1,\cdots,N\}\subseteq\VV$ \emph{service} agents that are deployed to provide the targets with a service, see Fig~\ref{fig::monitoring}. Unlike some existing literature like~\cite{SF-MM-ZY-GBW:15}, $\VV_{\textup{a}}$ and $\VV_{\textup{s}}$ do not have to be mutually exclusive. We assume that the active agents have partitioned the area such that each target is detected only by one active agent, i.e., no overlapping~detection.

We let $(\vect{x}^i_{\textup{s}},\theta^i_{\textup{s}})\in\real^2\times[0,2\uppi]$ be the pose (position and orientation) of service agent $i\in\VV_{\textup{s}}$. 
The QoS distribution $Q^i(\vect{x}|\vect{x}^i_{\textup{s}},\theta^i_{\textup{s}})$ provided by a service agent $i\in\VV_{\textup{s}}$ is modeled by a scaled Gaussian probability distribution $Q^i(\vect{x}|\vect{x}^i_{\textup{s}},\theta^i_{\textup{s}})=z^i \mathcal{N}(\vect{x}|\vect{x}^i_{\textup{s}},\Sigma(\theta^i_{\textup{s}}))$, $\vect{x}\in\real^2$, where  $z^i\in\realpositive$ is the scale constant  and
$\mathcal{N}(\vect{x}|\vect{x}^i_{\textup{s}},\Sigma(\theta^i_{\textup{s}}))$ is the Gaussian distribution with the mean $\vect{x}^i_{\textup{s}}$ and the covariance matrix $\Sigma(\theta^i_{\textup{s}})$.
We define the normalized collective QoS provided by the service agents by the probability density distribution
\begin{align}&q(\vect{x}|\{\vect{x}^i_{\textup{s}},\theta^i_{\textup{s}}\}_{i\in\VV_{\textup{s}}})=\frac{\sum_{i\in\VV_{\textup{s}}} Q^i}{\int_{\vect{x\in\real^2}}\sum_{i\in\VV_{\textup{s}}} Q^i d\vect{x}}\nonumber\\
&=\frac{\sum_{i\in\VV_{\textup{s}}} z^i \mathcal{N}(\vect{x}|\vect{x}^i_{\textup{s}},\Sigma(\theta^i_{\textup{s}}))}{\sum_{i\in\VV_{\textup{s}}} z^i}=\sum_{i\in\VV_{\textup{s}}} \omega^i_{\textup{s}} \mathcal{N}(\vect{x}|\vect{x}^i_{\textup{s}},\Sigma(\theta^i_{\textup{s}})),\end{align} where $\omega^i_{\textup{s}}=\frac{z^i}{\sum_{i\in\VV_{\textup{s}}} z^i}$ represents the relative service capability of agent $i$ among $\VV_{\textup{s}}$. 

\begin{figure}
  \centering
  \includegraphics[width=0.38\textwidth]{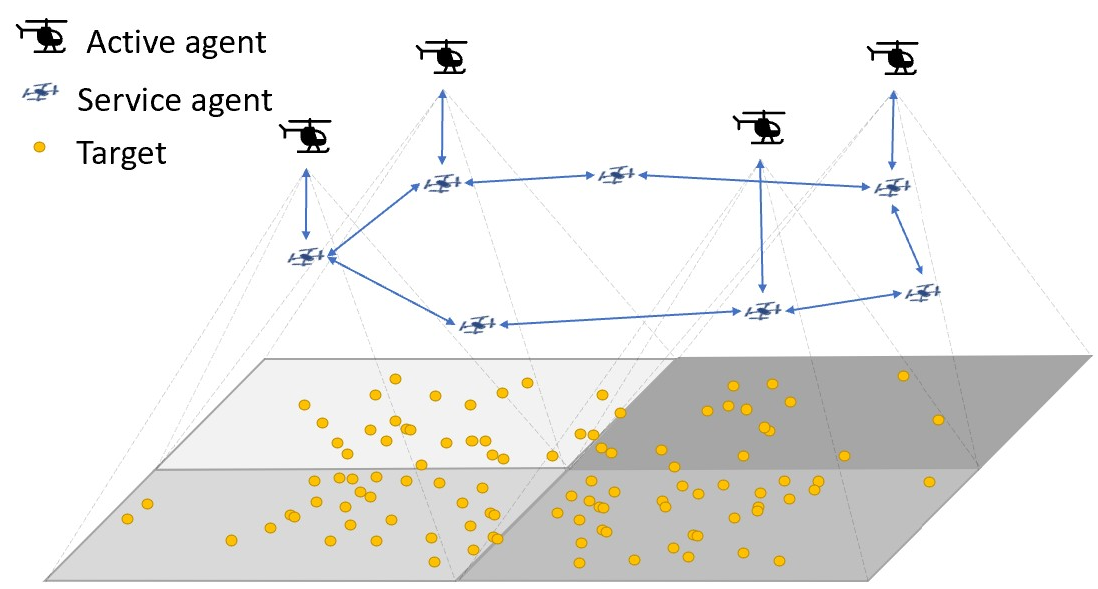}
  \caption{{\small A multi-agent system with active agents and service~agents.}}
  \label{fig::monitoring}
\end{figure}

Our objective in this paper is to first enable all the agents, both active and service agents, obtain an estimate mixture model $\hat{p}(\vect{x})$ of the density distribution of the targets in a distributed manner. Then, design a distributed deployment strategy to re-position the service agents in a way that their collective QoS serves the targets in an efficient manner. In other words, we seek locations and orientations for service agents such that the collective QoS distribution $q$ is as much similar to as possible to the estimated target density distribution $\hat{p}$. The optimal solution for the deployment objective can be obtained~from 
\begin{align}\label{eq::objective}
   \{\vect{x}_{\textup{s}}^i,\theta_{\textup{s}}^i\}_{i\in\VV_{\textup{s}}}= \arg\min\KLD\big(\hat{p}(\vect{x})||q(\vect{x})\big).
\end{align} 
We note that $\hat{p}(\vect{x})$ and $q(\vect{x})$ are mixture distributions for which obtaining a closed-from for their KLD is quite challenging. In practice, KLD for mixture models are usually estimated by using costly Monte-Carlo sampling simulations~\cite{HJR-OPA:07}. Moreover, the collective QoS distribution $q(\vect{x})$ contributed by each agent's QoS distribution, $\omega_{\textup{s}}^i \mathcal{N}(\vect{x}|\vect{x}^i_{\textup{s}},\Sigma(\theta^i_{\textup{s}}))$, $i\in\VV_{\textup{s}}$, is a global information. Accordingly, designing a distributed solver for~\eqref{eq::objective} is challenging. Therefore, in this paper, we seek a suboptimal solution for~\eqref{eq::objective} that can be implemented in a distributed manner and has low computational complexity.

\section{Overview of the Proposed Mobile Agent Deployment Solution}
Our proposed distributed solution to meet our objective stated in Section~\ref{sec::prob_def} is the two-stage process depicted in Fig.~\ref{fig::plan}. In the first stage, we use a GMM with $N$  Gaussian bases to model the target density distribution. The active agents $\VV_{\textup{a}}$ detect the positions of the targets, considered as the sampled data from the unknown distribution $p(\vect{x})$. Then, a distributed EM algorithm, which uses a set of active weighted average consensus algorithms, is used to enable both active and service agents obtain a coherent estimate of the parameters of the $N$  Gaussian bases of the GMM. The Gaussian bases of the GMM partition the target area into $N$ subregions each of which corresponds to a Gaussian basis.  The second stage of our solution is an agent allocation process which follows an optimal mass transport framework. In this allocation process,  first each service agent $i\in\mathcal{V}_{\textup{s}}$ computes the KLD between its QoS distribution, $\omega^i_{\textup{s}} \mathcal{N}(\vect{x}|\vect{x}^i_{\textup{s}},\Sigma(\theta^i_{\textup{s}}))$, and each subregion's Gaussian basis obtained in stage 1. Then, a distributed assignment problem is formulated with the KLDs as the cost of deploying the agent to the each respective subregion. As a result, each agent is paired with a subregion and the summation of the divergences corresponding to each paired agent's QoS distribution and subregion's Gaussian basis is minimized. The last step in this stage is a transportation process in which we implement a local controller to drive the agents to their assigned destinations in finite time. For dynamic targets, the process repeats. We present the details of each stage in the following sections.

\begin{figure}
  \centering
  \includegraphics[width=0.38\textwidth]{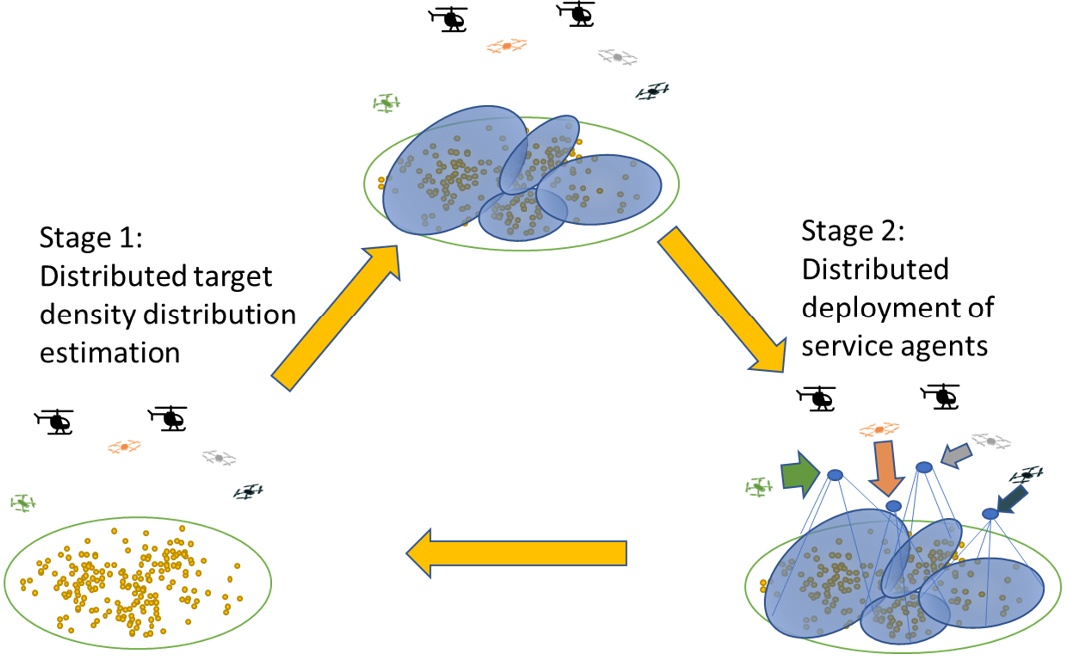}
  \caption{
  {\small The proposed two-stage distributed deployment solution.}  }
  \label{fig::plan}
\end{figure}

\section{Stage 1: distributed target density distribution estimation}\label{sec::stage1}


GMM is characterized by finite sum of Gaussian bases with different weights, means and covariance  matrices. Let $\vect{x}\in\real^2$ be the observed target's position drawn from a mixture of $N$ Gaussian bases with the distribution $\mathcal{N}(\vect{x}|\vect{\mu}_k,\vect{\Sigma}_k)$, where $\vect{\mu}_k\in\real^2$ is the mean and $\vect{\Sigma}_k\in\real^{2\times 2}$ is the covariance matrix for $k\in\mathcal{K}=\{1,\cdots,N\}$. Let $z\in\real$ be the indicator which indicates the variable $\vect{x}$ belongs to $k^\text{th}$ Gaussian basis when $z=k$. The variable $z$ is not observed so $z$ is also called hidden variable or latent variable. The probability of drawing a variable from the $k^\text{th}$ Gaussian basis is denoted as $\pi_k :=\Pr(z=k)$. The distribution of $\vect{x}$ given the $k^\text{th}$ mixture basis is Gaussian, i.e., $\hat{p}(\vect{x}|z=k)=\mathcal{N}(\vect{x}|\vect{\mu}_k,\vect{\Sigma}_k)
$. Therefore, the marginal probability distribution for $\vect{x}$ is given by
\begin{align}
\hat{p}(\vect{x})=\sum\nolimits_{k=1}^{N}\pi_k\,\mathcal{N}(\vect{x}|\vect{\mu}_k,\vect{\Sigma}_k)
\end{align}
The parameters that should be determined to obtain the estimate $\hat{p}(\vect{x})$ are the set $\{{\pi}_k,\vect{\mu}_k,\vect{\Sigma}_k\}_{k=1}^N$. Next, we employ the EM algorithm to obtain these parameters \cite{DAP-LNM-RDB:77}.

The EM algorithm obtains the maximum likelihood estimates of $\{\pi_k,\vect{\mu}_k,\vect{\Sigma}_k\}_{k=1}^N$ given $M$ independent detected targets' positions $\{\vect{x}^n_{\textup{t}}\}_{n=1}^M$. It is an iterative method that alternates between an expectation (E) step and a maximization (M) step. Given a detected target $\vect{x}^n_t$, $n\in\{1,\cdots,M\}$, E-step computes the posterior probability
\begin{align}\label{eq::gamma}
    \gamma_{kn}:=\, \Pr(z=k|\vect{x}^n_t)
    =\frac{\pi_k\,\mathcal{N}(\vect{x}^n_t|\vect{\mu}_k,\vect{\Sigma}_k)}{\sum_{j=1}^N \pi_j\,\mathcal{N}(\vect{x}^n_t|\vect{\mu}_j,\vect{\Sigma}_j)},
\end{align}
using the current value of $\{\pi_k,\vect{\mu}_k,\vect{\Sigma}_k\}_{k=1}^N$. 
Then, M-step updates the parameter set $\{\pi_k,\vect{\mu}_k,\vect{\Sigma}_k\}_{k=1}^N$ by the following equations using the current $\gamma_{kn}$:
\begin{subequations}\label{eq::EM_k}
\begin{align}
   \pi_k=&\frac{\sum_{n=1}^M\gamma_{kn}}{M},\label{eq::pi}\\
    \vect{\mu}_k=&\frac{\sum_{n=1}^M\gamma_{kn}\vect{x}^n_t}{\sum_{n=1}^M\gamma_{kn}},\label{eq::mu}  \\
    \vect{\Sigma}_k=&\frac{\sum_{n=1}^M\gamma_{kn}(\vect{x}^n_t-\vect{\mu}_k)(\vect{x}^n_t-\vect{\mu}_k)^\top}{\sum_{n=1}^M\gamma_{kn}},\label{eq::Sigma}
\end{align}
\end{subequations}
for $k\in\{1,\cdots,N\}$.  M-step needs the global information to update the parameter set $\{\pi_k,\vect{\mu}_k,\vect{\Sigma}_k\}_{k=1}^N$ because the summations in \eqref{eq::EM_k} are over all detected targets $n\in\{1,\cdots,M\}$. However, the information of the targets' positions $\{\vect{x}^n_t\}_{n=1}^M$ is distributed among the active agents $\VV_{\textup{a}}$. We observe that the right hand side quantities of~\eqref{eq::EM_k} are in the form of (weighted) average.  Therefore, we propose a distributed implementation of the EM algorithm, which invokes a set of active weighted average consensus algorithms such that all the agents, $\VV=\VV_{\textup{a}}\cup\VV_{\textup{s}}$ obtain an approximate value of~\eqref{eq::EM_k} by locally exchanging the information with their neighbors. Suppose each agent $i\in\VV$ maintains a local copy of the parameter set of the Gaussian bases $\{\pi_k^i,\vect{\mu}_k^i,\vect{\Sigma}_k^i\}_{k=1}^N$. At the E-step, every active agent $i\in\VV_
\textup{a}$ computes $\gamma_{kn}$ for $k\in\{1,\cdots,N\}$ and $n\in\VV^i_t$ where $\VV^i_t$ is the set of targets detected by active agent $i\in\VV_{\textup{a}}$. Then,
in the M-step, every agent $i\in\VV$ executes Consensus Algorithm \ref{alg::consensus} with proper setting its weight $\eta^i$ and  reference $\vect{\mathsf{r}}^i$ to estimate the update of $\{\pi_k^i,\vect{\mu}_k^i,\vect{\Sigma}_k^i\}_{k=1}^N$. 
It is clear that by setting $\eta^i=|\VV^i_t|$ and $\vect{\mathsf{r}}^i=\frac{\sum_{n\in\VV^i_t}\gamma_{kn}}{|\VV^i_t|}$ if $i\in\VV_{\textup{a}}$, otherwise, $\eta^i=0$ and $\vect{\mathsf{r}}^i=0$, the consensus variable $\vect{y}^i$ in Algorithm~\ref{alg::consensus} asymptotically converges to~\eqref{eq::pi}. Similarly, letting $\eta^i=\sum_{n\in\VV^i_t}\gamma_{kn}$ and $\vect{\mathsf{r}}^i=\frac{\sum_{n\in\VV^i_t}\gamma_{kn}\vect{x}_n}{\sum_{n\in\VV^i_t}\gamma_{kn}}$ if $i\in\VV_{\textup{a}}$, otherwise, $\eta^i=0$ and $\vect{\mathsf{r}}^i=0$, $\vect{y}^i$ converges to \eqref{eq::mu}; letting $\eta^i=\sum_{n\in\VV^i_t}\gamma_{kn}$ and $\vect{\mathsf{r}}^i=\frac{\sum_{n\in\VV^i_t}\gamma_{kn}(\vect{x}_n-\vect{\mu}^i_k)(\vect{x}_n-\vect{\mu}^i_k)^\top}{\sum_{n\in\VV^i_t}\gamma_{kn}}$ if $i\in\VV_{\textup{a}}$, otherwise, $\eta^i=0$ and $\vect{\mathsf{r}}^i=0$, $\vect{y}^i$ converges to \eqref{eq::Sigma}. 
The proposed consensus based distributed EM algorithm is summarized in Algorithm \ref{alg::distributed_EM}.

\begin{algorithm}[t]
\caption{{\small Consensus-based distributed EM algorithm for GMM}}\label{alg::distributed_EM}
\footnotesize
\begin{algorithmic}
\Require Detected targets set $\{\vect{x}^n_t\}_{n\in \VV^i_t}$ by agent $i$, number of Gaussian bases $N$, number of loops $T$\\
\textbf{Initialization:} $\{\pi_k^i,z^i_{\pi,k},v^i_{\pi,k}\}_{k=1}^N$,$\{\vect{\mu}_k^i,\vect{z}^i_{\mu,k},\vect{v}^i_{\mu,k}\}_{k=1}^N$,\\$\{\vect{\Sigma}_k^i,\vect{z}^i_{\Sigma,k},\vect{v}^i_{\Sigma,k}\}_{k=1}^N$. \\
\For{$t=1:T$}
\If{$i\in\VV_{\textup{a}}$} \Comment{E-step}
\\
Compute $\gamma_{kn}$ in~\eqref{eq::gamma} using the current value of $\{\pi_k^i\,\vect{\mu}_k^i,\vect{\Sigma}_k^i\}$ for $k=\{1,\cdots,N\}$ and $n\in\VV^i_t$.
\EndIf
\For{$k=1:N$}\Comment{M-step}
\If{$i\in\VV_{\textup{a}}$}
\begin{align*}
    [\pi_k^i,z^i_{\pi,k},v^i_{\pi,k}]\gets&\,\, \Con(|\VV^i_t|,\frac{\sum_{n\in\VV^i_t}\gamma_{kn}}{|\VV^i_t|},z^i_{\pi,k},v^i_{\pi,k})\\
    \!\!\!\![\vect{\mu}_k^i,\vect{z}^i_{\mu,k},\vect{v}^i_{\mu,k}]\gets&\,\, \Con(\sum\limits_{n\in\VV^i_t}\!\!\gamma_{kn},\frac{\sum\limits_{n\in\VV^i_t}\!\!\!\gamma_{kn}\vect{x}_n}{\sum\limits_{n\in\VV^i_t}\!\!\!\gamma_{kn}},\vect{z}^i_{\mu,k},\vect{v}^i_{\mu,k})\\
    [\vect{\Sigma}_k^i,\vect{z}^i_{\Sigma,k},\vect{v}^i_{\Sigma,k}]\gets&\\ &\!\!\!\!\!\!\!\!\!\!\!\!\!\!\!\!\!\!\!\!\!\!\!\!\!\!\!\!\!\!\!\!\!\!\!\!\!\!\!\!\!\!\!\!\!\!\!\!\!\Con(\sum_{n\in\VV^i_t}\!\!\gamma_{kn},\frac{\sum\limits_{n\in\VV^i_t}\gamma_{kn}(\vect{x}_n\!-\!\vect{\mu}^t_k)(\vect{x}_n\!-\!\vect{\mu}^t_k)^\top}{\sum\limits_{n\in\VV^i_t}\!\gamma_{kn}},\vect{z}^i_{\Sigma,k},\vect{v}^i_{\Sigma,k})
\end{align*}

\Else
\begin{align*}
    [\pi_k^i,z^i_{\pi,k},v^i_{\pi,k}]&\gets \Con(0,0,z^i_{\pi,k},v^i_{\pi,k})\\
    [\vect{\mu}_k^i,\vect{z}^i_{\mu,k},\vect{v}^i_{\mu,k}]&\gets \Con(0,0,\vect{z}^i_{\mu,k},\vect{v}^i_{\mu,k})\\
        [\vect{\Sigma}_k^i,\vect{z}^i_{\Sigma,k},\vect{v}^i_{\Sigma,k}]&\gets \Con(0,0,\vect{z}^i_{\Sigma,k},\vect{v}^i_{\Sigma,k})
\end{align*}
\EndIf \EndFor \EndFor\\
\Return $\{\pi_k^i,\vect{\mu}_k^i,\vect{\Sigma}_k^i\}_{k=1}^N$
\end{algorithmic}
\end{algorithm}
Because the algorithm is terminated in a finite time. It is expected that $\hat{p}^i(x)$ of each agent $i$ be slightly different than other agents. In what follows we let,
\begin{align}\label{eq::hatpi}
\hat{p}^i(\vect{x})=\sum\nolimits_{k=1}^{N}\pi^i_k\,\mathcal{N}(\vect{x}|\vect{\mu}^i_k,\vect{\Sigma}^i_k),
\end{align}
be the local final estimate of agent $i\in\mathcal{V}$.

\subsection{Numerical demonstration}\label{sec::Stage 1_demonstration}
We demonstrate a numerical simulation to show the performance of the proposed distributed EM algorithm. Consider a group of $6$ mobile agents where $\mathcal{V}_{\textup{a}}=\{1,2,3,6\}$ are the active agents that monitor the targets to enable the service agents $\mathcal{V}_{\textup{s}}=\mathcal{V}=\{1,2,3,4,5,6\}$ to obtain an estimate of the density distribution ${p}(x)$ of 
 a group of $M=1000$ targets. The agents $\mathcal{V}$ communicate over a connected ring graph whose adjacency matrix is
\begin{align*}\vect{\sf{A}}=\left[\begin{smallmatrix}
0 & -1 & 0 & 0 & 0 & -1 \\
-1 & 0 & -1  &  0 & 0 & 0\\
0 & -1 & 0 & -1 & 0 & -1 \\
0 & 0 & -1 & 0 & -1 & 0 \\
0 & 0 &  0 & -1 & 0 & -1 \\
-1 & 0 & -1 & 0 & -1 & 0 \end{smallmatrix}\right].
\end{align*}
The numbers of the targets detected by agent $i\in\VV_{\textup{a}}$ are $|\VV^1_{\textup{t}}|=100$, $|\VV^2_{\textup{t}}|=250$, $|\VV^3_{\textup{t}}|=450$, and $|\VV^6_{\textup{t}}|=200$. The agents, both active and service, execute the distributed EM Algorithm~\ref{alg::distributed_EM} to estimate the parameters of the Gaussian bases of GMM for the target density distribution. In the simulation, the number of the iteration-loops of the consensus algorithm and EM algorithm are $L=20$ and $T=50$, respectively. 
Agents' estimation results are illustrated in Fig. \ref{fig::ConsensusEMAgent} where the black circle's represent the targets, the elliptic footprints are the 3-$\sigma$ uncertainty ellipses of the $6$ Gaussian bases $\mathcal{N}(\vect{x}|\vect{\mu}_k^i,\vect{\Sigma}_k^i), k\in\{1,\cdots,6\}$ estimated by agent $i$ and the thickness of the elliptic footprint represents $\pi_k^i$. The result shows that with the proposed distributed EM algorithm, the agents successfully estimate the parameter of GMM for the target's distribution though agent $4$ and $5$ do not detect any target.
We note that with the help of the consensus algorithm all agents get the approximately same estimation~results.


\begin{figure}
  \centering
  \includegraphics[width=0.38\textwidth]{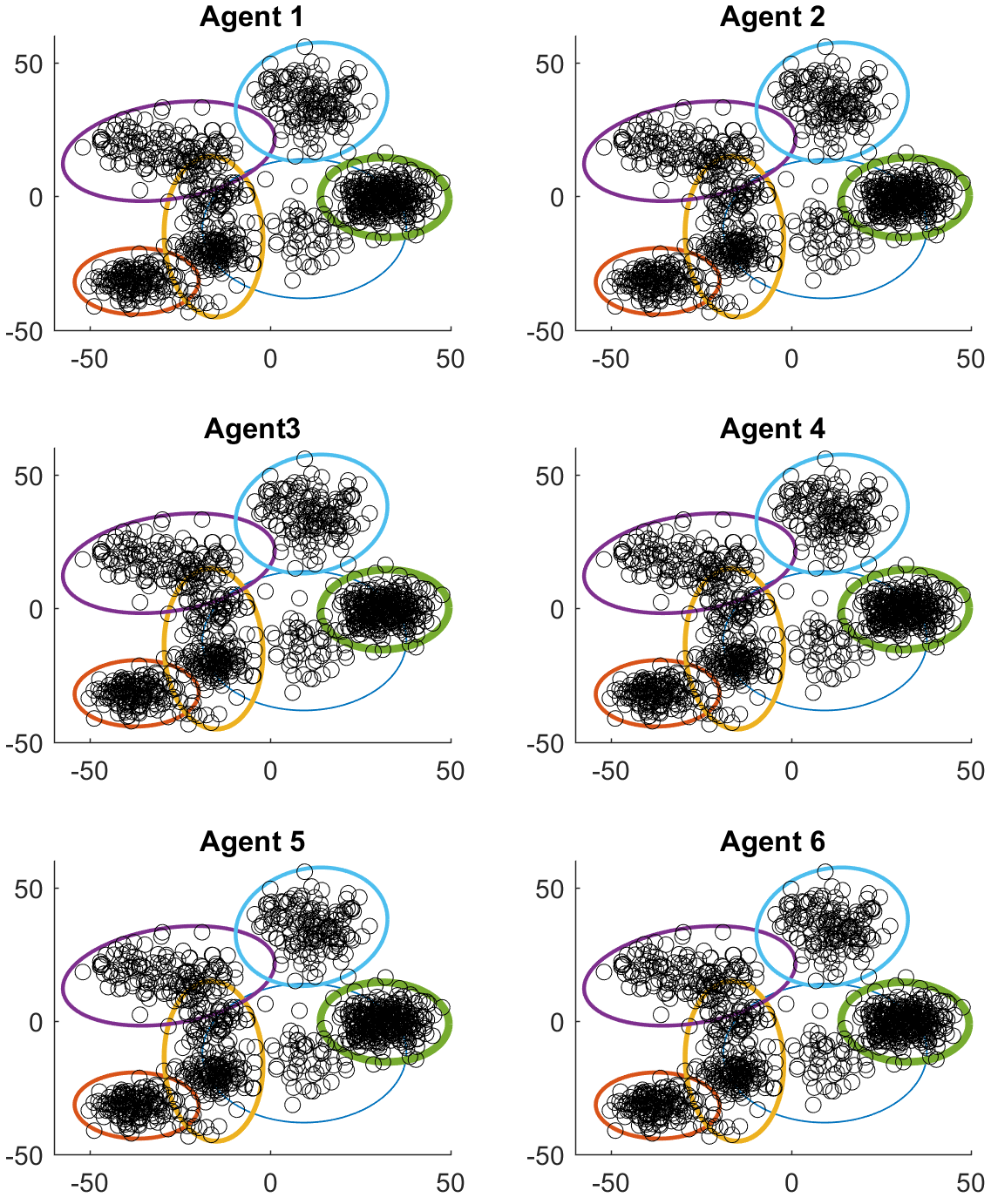}
  \caption{{\small The estimate of GMM of each agents of the demonstration in section \ref{sec::Stage 1_demonstration}.}}
  \label{fig::ConsensusEMAgent}
\end{figure}


\section{Stage 2: Distributed Deployment of Service Agents}\label{sec::stage 2}
In stage 1, the target density distribution is modeled and estimated by a GMM. The result of GMM intrinsically partitions the area into a set of subregions each of which represents a Gaussian basis.
Our suboptimal solution to the deployment problem~\eqref{eq::objective} is to deploy each service agent $i\in\VV_{\textup{s}}=\{1,\cdots,N\}$ to optimally cover an assigned subregion $k\in\mathcal{K}=\{1,\cdots,N\}$. The service agent assignment is based on the similarity  of the agent's QoS distribution, $\omega_{\textup{s}}^i \mathcal{N}(\vect{x}|\vect{x}^i_{\textup{s}},\Sigma(\theta^i_{\textup{s}}))$, to the  Gaussian basis subregion, $\pi_k^i \mathcal{N}(\vect{x}|\vect{\mu}_k^i,\vect{\Sigma}_k^i)$, such that the summation of the KLD of each assigned agent-subregion pair is minimized. This objective can be formalized as follows. For any service agent $i\in\VV_{\textup{s}}$ let 
\begin{align}\label{eq::cost_KLD}
&C_{ik}(\vect{x}_{\textup{s}}^i,\theta_{\textup{s}}^i)=\KLD\big(\pi_k^i \mathcal{N}(\vect{x}|\vect{\mu}_k^i,\vect{\Sigma}_k^i)||\omega_{\textup{s}}^i \mathcal{N}(\vect{x}|\vect{x}^i_{\textup{s}},\Sigma(\theta^i_{\textup{s}}))\nonumber\\
=&\pi_k^i\Big(\ln\frac{\pi_k^i}{\omega_{\textup{s}}^i}+\KLD\big( \mathcal{N}(\vect{x}|\vect{\mu}_k^i,\vect{\Sigma}_k^i)|| \mathcal{N}(\vect{x}|\vect{x}^i_{\textup{s}},\Sigma(\theta^i_{\textup{s}}))\big)\Big),
\end{align}
for $k\!\in\!\mathcal{K}$. 
We note that $C_{ik}$ in~\eqref{eq::cost_KLD} is a continuous function of the service agent's pose $(\vect{x}_{\textup{s}}^i,\theta_{\textup{s}}^i)$. We introduce a binary decision variable $Z_{ik}\in\{0,1\}$, which is $1$ if agent $i$ is assigned to region $k$ and $0$ otherwise. With the right notation at hand then, our suboptimal deployment solution is given~by
\begin{align}\label{eq::sub_op}
\{\vect{x}_{\textup{s}}^{i^\star},&\theta_{\textup{s}}^{i^\star},\{Z_{ik}^\star\}_{k\in\mathcal{K}}\}_{i\in \VV_s}= \arg\min \sum_{i\in\mathcal{V}_{\textup{s}}}\sum_{k\in\mathcal{K}}C_{ik}(\vect{x}_{\textup{s}}^i,\theta_{\textup{s}}^i)Z_{ik},~~\\
    &Z_{ik}\in\{0,1\},\quad i\in\VV_{\textup{s}},~k\in\mathcal{K},\nonumber\\
     &\sum_{k\in\mathcal{K}}Z_{ik}=1,\qquad \forall i\in\VV_{\textup{s}},\nonumber\\
    &\sum_{i\in\VV_{\textup{s}}}Z_{ik}=1, \qquad \forall k\in\mathcal{K}.\nonumber
\end{align}
Next, we introduce a set of manipulations that allows us to arrive at a distributed solution for solving~\eqref{eq::sub_op}. For each service agent $i\in\mathcal{V}_{\textup{s}}$, we start by defining
\begin{align}\label{eq::Cstar}
    C_{ik}^\star=\underset{\vect{x}_{\textup{s}}^i,\theta_{\textup{s}}^i}{\min} C_{ik}(\vect{x}_{\textup{s}}^i,\theta_{\textup{s}}^i),~~k\in\mathcal{K}.
\end{align}
Given~\eqref{eq::Cstar} and the fact that $C_{ik}$ depends on the pose of agent $i$ only, it is straightforward to show that~\eqref{eq::sub_op} can be written in the equivalent form of 
\begin{align}\label{eq::sub_op_eqv}
Z_{ik}^\star&= \arg\min\sum_{i\in\mathcal{V}_{\textup{s}}}\sum_{k\in\mathcal{K}}C^\star_{ik}Z_{ik},~~\\
    &Z_{ik}\in\{0,1\},\quad i\in\VV_{\textup{s}},~k\in\mathcal{K},\nonumber\\
    &\sum_{k\in\mathcal{K}}Z_{ik}=1,\qquad \forall i\in\VV_{\textup{s}},\nonumber\\
    &\sum_{i\in\VV_{\textup{s}}}Z_{ik}=1, \qquad \forall k\in\mathcal{K}.\nonumber
\end{align}
where $(\vect{x}_{\textup{s}}^{i^\star},\theta_{\textup{s}}^{i^\star})$ for each service agent $i\in\mathcal{V}_{\textup{s}}$ is equal to  minimizer $(\vect{x}_{\textup{s}}^{ik^\star},\theta_{\textup{s}}^{ik^\star})$ of the $k$th~\eqref{eq::Cstar} that corresponds to $Z_{ik}^\star=1$.
 The equivalent optimization representation~\eqref{eq::sub_op_eqv} casts our suboptimal service agent assignment problem in the form of a discrete optimal mass transport problem \cite{PG-CM:19}. In this optimal mass transport problem, the minimum value of~\eqref{eq::cost_KLD} given in~\eqref{eq::Cstar} can be viewed as the cost of assigning agent $i$ to the $k$th  subregion/basis of the GMM. In Section~\ref{sec::distributed_task_assignment}, we show that the mixed integer programming problem~\eqref{eq::sub_op_eqv}, in fact can be cast as a linear programming in continuous space, and then solved in a distributed manner using existing optimization algorithms. Once each service agent obtains its assigned pose, we transport the agents to their assigned region by a finite-time minimum energy control. In what follows, before presenting our equivalent linear programming representation of~\eqref{eq::sub_op_eqv}, we discuss how we can obtain the minimizers of~\eqref{eq::Cstar}. More specifically, in Section~\ref{sec::stage 2_Gaussian} we show that the minimum value for each $C_{ik}$ happens at $\vect{x}_{\textup{s}}^{ik^\star}=\vect{\mu}_k$ and the orientation $\theta_{\textup{s}}^{ik^\star}$ that makes the principal axis of the uncertainty ellipses of the service distribution and the corresponding Gaussian distribution are in parallel.

\subsection{Similarity assessment for Gaussian QoS distribution to Gaussian basis subregion}\label{sec::stage 2_Gaussian}
Given the QoS distribution provided by agent $i\in\VV_{\textup{s}}$ to be $\omega_{\textup{s}}^i\mathcal{N}(\vect{x}|\vect{\vect{x}}_{\textup{s}}^i,\vect{\Sigma}^i(\theta_{\textup{s}}^i))$, where the mean of the Gaussian distribution is at the agent's location $\vect{x}_{\textup{s}}^i$ and the covariance matrix is with principal (major) axis at angle $\theta_{\textup{s}}^i$, see Fig.~\ref{fig::Gaussian_axis}. Hence, the covariance matrix can be decomposed into $\vect{\Sigma}^i(\theta_{\textup{s}}^i)=\vect{R}(\theta_{\textup{s}}^i)\vect{\Lambda}^i\vect{R}(\theta_{\textup{s}}^i)^\top$, where $\vect{R}(\theta_{\textup{s}}^i)=
\left[\begin{smallmatrix}
\cos{\theta_{\textup{s}}^i} & -\sin{\theta_{\textup{s}}^i}\\
\sin{\theta_{\textup{s}}^i} & \cos{\theta_{\textup{s}}^i}
\end{smallmatrix}\right]$ and $\vect{\Lambda}^i=\left[\begin{smallmatrix}
\sigma_x^i & 0\\
0 & \sigma_y^i
\end{smallmatrix}\right]$, in which $\sigma_x^i,\sigma_y^i\in\realpositive$ with $\sigma_x^i\geq\sigma_y^i$ are known service parameters  determines the `shape' of the service agent $i$. Similarly, agent $i$'s estimated covariance matrix $\vect{\Sigma}_k^i$, for the $k$th subregion/basis of its estimated $\hat{p}(x)$, see~\eqref{eq::hatpi}, can be written as $\vect{\Sigma}^i_k(\theta_k^i)=\vect{R}(\theta^i_k)\vect{\Lambda}_k^i\vect{R}(\theta^i_k)^\top$, where $\vect{\Lambda}^i_k=\left[\begin{smallmatrix}
\sigma^i_{k,x} & 0\\
0 & \sigma^i_{k,y}
\end{smallmatrix}\right]$, in which $\theta_k^i$ is the angle of principal (major) axis of the covariance matrix and $\sigma_{k,x}^i,\sigma_{k,y}^i\in\realpositive$ with $\sigma_{k,x}^i\geq\sigma_{k,y}^i$ are the variances in the major axis and minor axis direction, respectively, see Fig.~\ref{fig::Gaussian_axis}. 
With the right notation at hand, the theorem below gives a closed-form solution for the minimizer ($\vect{x}_{\textup{s}}^{ik^\star},\theta_{\textup{s}}^{ik^\star}$) of~\eqref{eq::Cstar}.

\begin{figure}
  \centering
  \includegraphics[width=0.35\textwidth]{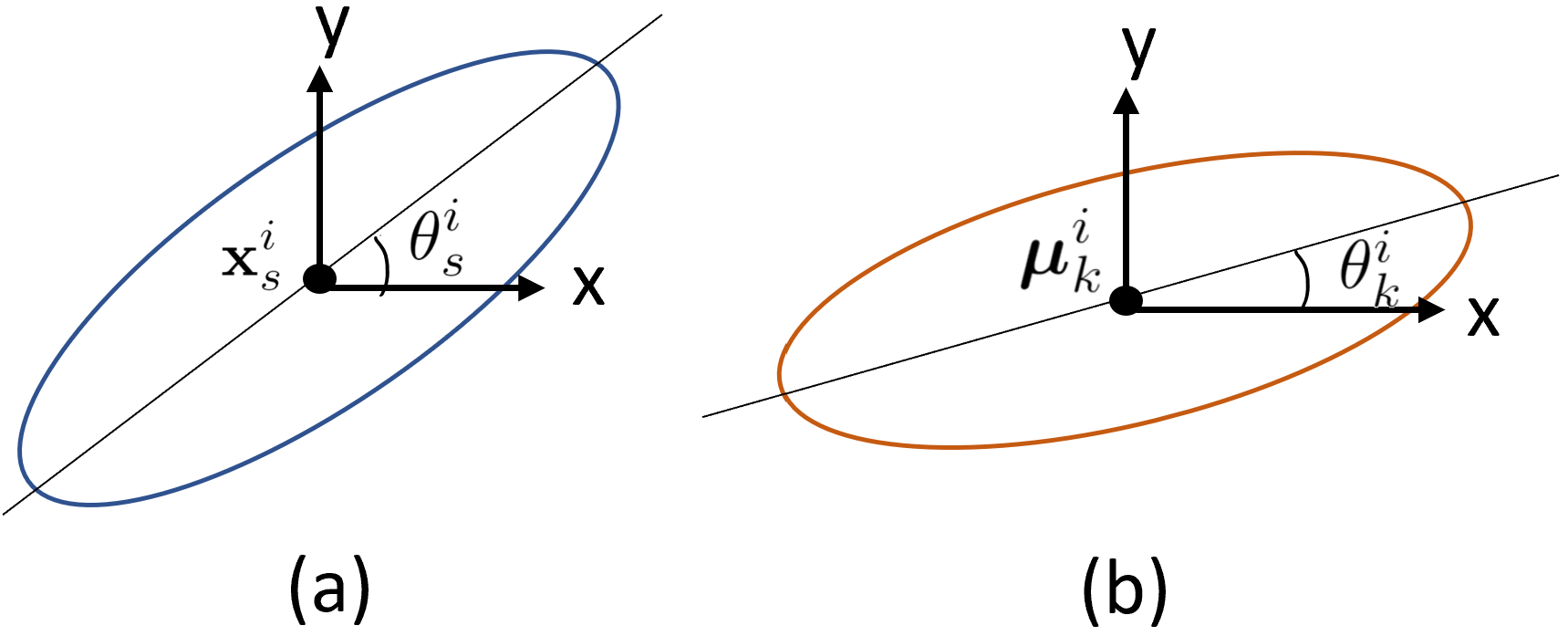}
  \caption{{\small The principal axis angle of (a) agent $i$'s QoS Gaussian distribution and (b) the $k$the subregion/basis of $\hat{p}^i(x)$.}}
  \label{fig::Gaussian_axis}
\end{figure}

\begin{thm}\label{thm::guassian_allocate}
Consider the optimization problem~\eqref{eq::Cstar}.
Then, one of the global minimizer of optimization \eqref{eq::Cstar} is $(\vect{x}_{\textup{s}}^{ik^\star},\theta_{\textup{s}}^{ik^\star})=$ $(\vect{\mu}^i_k,\theta^i_k)$, where $\theta^i_k$ is the angle of the principal axis of $\vect{\Sigma}^i_k$.
Moreover,
\begin{align}\label{eq::minimum_KLD}
  C_{ik}^\star=\pi_k^i\Big(\ln\frac{\pi_k^i}{\omega_{\textup{s}}^i}\!+\!\frac{1}{2}\big(\ln\frac{\sigma_x^i \sigma_y^i}{\sigma_{k,x}^i\, \sigma_{k,y}^i}\!+\!\frac{\sigma_{k,x}^i\sigma_y^i\!+\!\sigma_{k,y}^i\sigma_x^i}{\sigma_x^i \sigma_y^i}\!-2\big)\Big).  
\end{align}
\end{thm}

\begin{proof}
We first note that since $\pi_k^i$ and $\omega_{\textup{s}}^i$ are fixed parameters, \eqref{eq::Cstar} is equivalent to minimize $\KLD\big( \mathcal{N}(\vect{x}|\vect{\mu}_k^i,\vect{\Sigma}_k^i)|| \mathcal{N}(\vect{x}|\vect{\vect{x}}_{\textup{s}}^i,\vect{\Sigma}^i(\theta_{\textup{s}}^i))\big)$. 
Invoking~\eqref{eq::close_KLD} we can write
\begin{align}\label{eq::KLD_twoGaussian}
&\KLD\big( \mathcal{N}(\vect{x}|\vect{\mu}_k^i,\vect{\Sigma}_k^i)|| \mathcal{N}(\vect{x}|\vect{\vect{x}}_{\textup{s}}^i,\vect{\Sigma}^i(\theta_{\textup{s}}^i))\big)\Big)\nonumber\\
&=\frac{1}{2}\big(\underbrace{\ln\frac{|\vect{\Sigma}^i(\theta_{\textup{s}}^i)|}{|\vect{\Sigma}_k^i|}}_{(a)}+\underbrace{(\vect{x}_{\textup{s}}^i-\vect{\mu}_k^i)^\top\vect{\Sigma}^i(\theta_{\textup{s}}^i)^{-1}(\vect{x}_{\textup{s}}^i-\vect{\mu}_k^i)}_{(b)}\nonumber\\
&\qquad+\underbrace{\tr(\vect{\Sigma}^i(\theta_{\textup{s}}^i)^{-1}\vect{\Sigma}_k^i)}_{(c)}-2 \big).
\end{align}
We note that, in~\eqref{eq::KLD_twoGaussian},
\begin{align*}
(a)=\ln\frac{|\vect{R}(\theta_{\textup{s}}^i)\vect{\Lambda}^i\vect{R}^\top(\theta_{\textup{s}}^i)|}{|\vect{R}(\theta_k^i)\vect{\Lambda}_k\vect{R}^\top(\theta_k^i)|}\!\!=\ln\frac{|\vect{\Lambda}^i|}{|\vect{\Lambda}_k|}\!\!=\ln\frac{\sigma_x^i \sigma_y^i}{\sigma_{k,x}^i\, \sigma_{k,y}^i},
\end{align*}
thus $(a)$ is a fix term and does not depend on the decision variable $\theta_{\textup{s}}^{i}$. Next, we note that $(b)$ is the only term in~\eqref{eq::KLD_twoGaussian} that depends on $\vect{\mu}_k^i$ and $\vect{x}_{\textup{s}}^i$. For any value other than ${\vect{x}_{\textup{s}}^i}=\vect{\mu}_k^i$, $(b)$ returns a positive value, which means that the minimum of~\eqref{eq::KLD_twoGaussian} happens at $\vect{x}_{\textup{s}}^{ik^\star}=\vect{\mu}_k^i$.  Lastly, we note that $(c)$ in \eqref{eq::KLD_twoGaussian} reads also as \begin{align*}
    (c)&=\tr(\vect{R}(\theta_{\textup{s}}^i)(\vect{\Lambda}^i)^{-1}\vect{R}(-\theta_{\textup{s}}^i+\theta_k^i)\vect{\Lambda}_k\vect{R}(\theta_k^i))\\
    &=\tr(\vect{R}(\theta_{\textup{s}}^i-\theta_k^i)(\vect{\Lambda}^i)^{-1}\vect{R}(-\theta_{\textup{s}}^i+\theta_k^i)\vect{\Lambda}_k).
\end{align*}
Now, let $\bar{\theta}=\theta_{\textup{s}}^i-\theta_k^i$, $\text{s}\bar{\theta}=\sin({\bar{\theta}})$ and $\text{c}\bar{\theta}=\cos({\bar{\theta}})$. Then, we can write (c) as
\begin{align*}
    (c)&=\tr\Big(\begin{bmatrix}
\text{c}\bar{\theta} & -\text{s}\bar{\theta}\\
\text{s}\bar{\theta} & \text{c}\bar{\theta}
\end{bmatrix}
\begin{bmatrix}
\frac{1}{\sigma_x^i} & 0\\
0 & \frac{1}{\sigma_y^i}
\end{bmatrix}
\begin{bmatrix}
\text{c}\bar{\theta} & \text{s}\bar{\theta}\\
-\text{s}\bar{\theta} & \text{c}\bar{\theta}
\end{bmatrix}
\begin{bmatrix}
\sigma_{k,x}^i & 0\\
0 & \sigma_{k,y}^i
\end{bmatrix}\Big)\\
&=\frac{(\sigma_{k,x}^i\sigma_y^i+\sigma_{k,y}^i\sigma_x^i)\text{c}^2\bar{\theta}+(\sigma_{k,x}^i\sigma_x^i+\sigma_{k,y}^i\sigma_y^i)\text{s}^2\bar{\theta}}{\sigma_x^i\sigma_y^i}.
\end{align*}
Let $\alpha=\sigma_{k,x}^i\sigma_y^i+\sigma_{k,y}^i\sigma_x^i$ and $\beta=\sigma_{k,x}^i\sigma_x^i+\sigma_{k,y}^i\sigma_y^i$. Then, $(c)$  reduces 
\begin{align*}
    (c)=\frac{\alpha+(\beta-\alpha)\text{s}^2\bar{\theta}}{\sigma_x^i\sigma_y^i}.
\end{align*}
Because $\sigma_{k,x}^i\geq\sigma_{k,y}^i$ and $\sigma_x^i\geq\sigma_y^i$, we have $\beta\geq\alpha$ and $(\beta-\alpha)\text{s}^2\bar{\theta}$ is non-negative. Hence, the global minimum of $(c)$ is $\frac{\alpha}{\sigma_x^i\sigma_y^i}=\frac{\sigma_{k,x}^i\sigma_y^i+\sigma_{k,y}^i\sigma_x^i}{\sigma_x^i \sigma_y^i}$ which happens at $\bar{\theta}^\star=n\uppi, n\in\{0,1,\cdots\}$, i.e., $\theta_{\textup{s}}^{ik^\star}=\theta_k^i+n\uppi, n\in\{0,1,\cdots\}$. To complete the proof, we note that $n=0$ leads to one of the global minimums $\theta_{\textup{s}}^{ik^\star}=\theta_k^i$.
\end{proof}
Given Theorem~\ref{thm::guassian_allocate}, if the optimization problem~\eqref{eq::sub_op_eqv} allocates service agent $i$ to the $k$th subregion/basis of $\hat{p}^i(x)$, the corresponding final pose of agent $i$ will be  $\vect{x}_{\textup{s}}^{i^\star}=\vect{\mu}_k^i$,  $\theta_{\textup{s}}^{i^\star}=\theta_k^i$.

\subsection{Distributed multi-agent assignment problem}\label{sec::distributed_task_assignment}
The assignment optimization problem~\eqref{eq::sub_op_eqv} is an integer optimization problem. As it is known in the discrete optimal mass transport literature \cite{PG-CM:19}, by the convex relaxation \cite{BDP:98}, the integer optimization~\eqref{eq::sub_op_eqv} can be transferred to the linear programming problem stated as follows:
\begin{align}\label{eq::linear_programming}
    &\min_{Z_{ik}\geq 0}\sum_{i\in\VV_{\textup{s}}}\sum_{k\in\mathcal{K}} C_{ik}^\star Z_{ik}\\
    &\textup{s.t.}\qquad \sum_{k\in\mathcal{K}}Z_{ik}=1,\qquad \forall i\in\VV_{\textup{s}},\nonumber\\
    &\qquad\quad \sum_{i\in\VV_{\textup{s}}}Z_{ik}=1, \qquad \forall k\in\mathcal{K}.\nonumber
\end{align}
Since only agent $i$ knows its own cost $C_{ik}^\star$ for $k\in\mathcal{K}$, we are interested in solving optimization problem \eqref{eq::linear_programming} in a distributed way. In general, problem \eqref{eq::linear_programming} may exist several optimal solutions ${Z_{ik}}^\star$. We also require the agents agreed on the same optimal assignment plan. A distributed simplex algorithm proposed by \cite{BM-NG-BF-AF:12} can achieve this aim. We rewrite \eqref{eq::linear_programming} to the standard form of linear programming
\begin{align}\label{eq::linear_programming_standard}
    &\min_{\vect{Z}} \vect{C}^{\star^T}\vect{Z}\\
    &\text{s.t.} \quad \vect{A}\vect{Z}=\vect{b},\quad \vect{Z}\geq 0.\nonumber
\end{align}
where $\vect{b}=\vect{1}_{2N}$,
\begin{align*}
    \vect{Z}=[Z_{11},\cdots,Z_{1N},Z_{21},\cdots,Z_{2,N},\cdots,Z_{N1}\cdots,Z_{NN}]^\top,\\
    \vect{C}^\star=[C_{11}^\star,\cdots,C_{1N}^\star,C_{21}^\star,\cdots,C_{2,N}^\star,\cdots,C_{N1}^\star\cdots,C_{NN}^\star]^\top,\\
    \vect{A}=[\vect{A}_{11},\cdots,\vect{A}_{1N},\vect{A}_{21},\cdots,\vect{A}_{2,N},\cdots,\vect{A}_{N1}\cdots,\vect{A}_{NN}],
\end{align*}
in which, $\vect{A}_{ik}\in\real^{2N}$ is a column vector with $i$-th and $(N+k)$-th entries are $1$, and others are $0$. A column of problem \eqref{eq::linear_programming_standard} is a vector $\vect{h}_{ik}\in\real^{1+2N}$ defined as $\vect{h}_{ik}=[C_{ik}^\star\quad \vect{A}_{ik}^\top]^\top$. The set of all columns is denote by $\mathcal{H}=\{\vect{h}_{ik}\}_{i\in\VV_{\textup{s}},k\in\mathcal{K}}$. Thus, the linear program \eqref{eq::linear_programming_standard} is fully characterized by the pair $(\mathcal{H},\vect{b})$. The information of $\mathcal{H}$ is distributed in the service agents. Let $\mathcal{P}^i=\{\vect{h}_{ik}\}_{k\in\mathcal{K}}$ is the problem column set known by agent $i\in\VV_\textup{s}$, which satisfies $\mathcal{H}=\cup_{i=1}^N \mathcal{P}^i$ and $\mathcal{P}^i\cap\mathcal{P}^j=\emptyset,\forall (i,j)\in\VV_{\textup{s}}$. We assume the communication graph $\mathcal{\GG}_{\textup{s}}(\VV_{\textup{s}},\EE_{\textup{s}})$ of the service agents is connected. Hence the tuple $(\mathcal{\GG}_{\textup{s}},(\mathcal{H},\vect{b}),\{\mathcal{P}^i\}_{i\in\VV_{\textup{s}}})$ forms a distributed linear program that can be solved by the distributed simplex algorithm \cite{BM-NG-BF-AF:12}. The result of the optimization problem \eqref{eq::linear_programming} is the optimal plan  $Z_{ik}^\star$, where $Z_{ik}^\star=1$ means assigning the agent $i$ to the $k$th subregion with the optimal pose ${\vect{x}_{\textup{s}}^i}^\star={\vect{x}_{\textup{s}}^{ik}}^\star$ and ${\theta_{\textup{s}}^i}^\star={\theta_{\textup{s}}^{ik}}^\star$.

The last step in Stage 2 of our deployment solution is agents transportation to their corresponding assigned pose. In practice, local controllers are expected to complete this task. One such local controller can be 
the well-known minimum energy control \cite[page 138]{FL-DV-VS:12} that can transport the agents to their respective assigned pose in finite time $\tau\in\real_{>0}$ while also enabling the agents to save on transportation energy. Let the local dynamics of agent $i$ (linearized) be given by $\dvect{\chi}^i(t)=\vect{A}^i\,\vect{\chi}^i(t)+\vect{B}^i\,\vect{u}^i(t)$, where $\vect{u}^i(t)$ is the control vector, and $\vect{\chi}^i$ is the state vector of agent $i$, which contains the pose and possibly other states. We assume that $(\vect{A}^i,\vect{B}^i)$ is controllable. 
Starting at an initial condition $\vect{\chi}^i(t_0)$, the minimum energy control is given by 
\begin{align}\label{eq::min_energy_control}
      \vect{u}^i(t)=\vect{B}^{i^\top}\! \ee^{\vect{A}^{i^\top}(t_0+ \tau-t)}\vect{G}^{i^{-1}}({\vect{\chi}^i}^\star-&\ee^{\vect{A}^i  \tau}\vect{\chi}^{i}(t_0))
    \end{align}  
     for $t\in[t_0,t_0+ \tau]$, where  $\vect{\chi}^{i\star}$ is agent $i$'s desired state whose pose component is  set to $({\vect{x}_{\textup{s}}^i}^\star, {\theta_{\textup{s}}^i}^\star)$, and $\vect{G}^i$ is the controllability Gramian.
Control \eqref{eq::min_energy_control} is a finite time control that guarantees to drive the agent from it initial state $\vect{\chi}^i(t_0)$
to it desired state ${\vect{\chi}^i}^\star$ 
in finite transportation time $\tau$, i.e., $\vect{\chi}^i(t_0 +\tau)=\vect{\chi}^{i^\star}$. Stage 2 of our distributed deployment is finished at $t_0+\tau$. If the targets are dynamic, our two-stage deployment process can repeat to re-position the service agents in accordance to the changes in targets distribution.

\section{Demonstrations}
\begin{figure}
  \centering
  \includegraphics[width=0.3\textwidth]{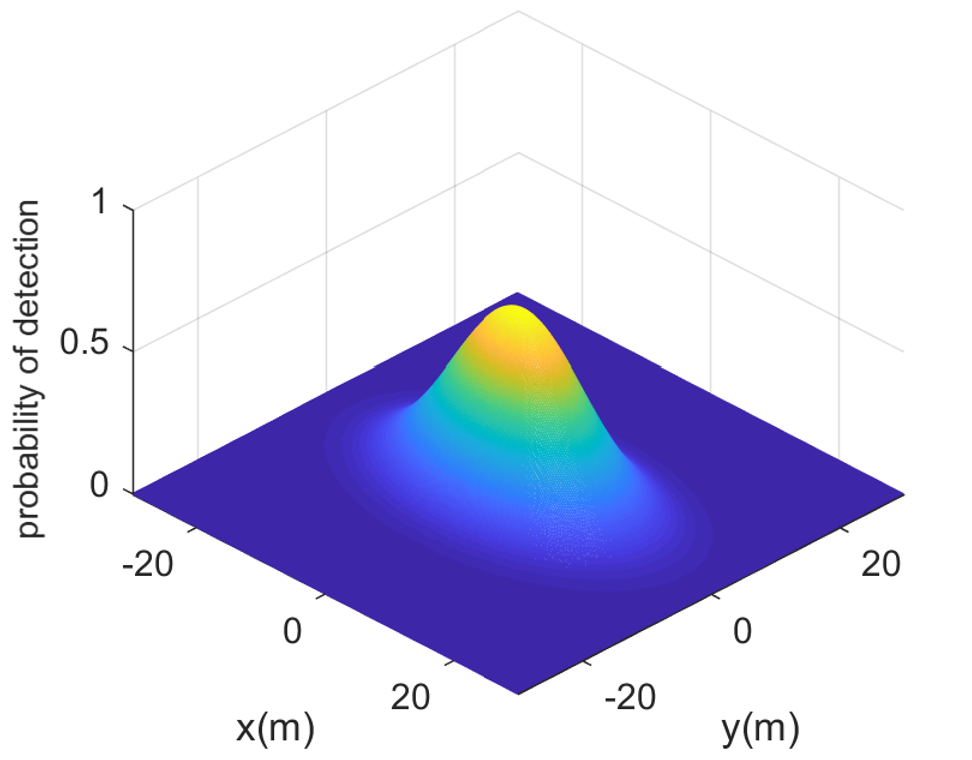}
  \caption{{\small The distribution of QoS of agent $1$}}
  \label{fig::QoS_sensor}
\end{figure}

Consider the the same setting in~Section~\ref{sec::Stage 1_demonstration} where a group of $6$ agents with the connected communication graph $\GG_{\textup{s}}$ used a distributed EM Algorithm~\ref{alg::consensus} to obtain the parameters $\{\pi_k^i\,\vect{\mu}_k^i,\vect{\Sigma}_k^i\}_{k=1}^6$ of the Gaussian bases of the targets distribution. The $6$ agents are equipped with a wireless sensor which is used to detect events of interest that occurred with targets. A commonly used sensor model is a probabilistic function conditioned on the sensor location and the event location~\cite{BF-FT-DWHF:03,XH-CR-XD:17}, i.e.,$\Pr(\text{Detected}|\vect{x}_{\textup{s}}^i,\vect{x}_t)$. For example in~\cite{XH-CR-XD:17}, given a sensor location at $\vect{x}_{\textup{s}}^i,i\in\VV_{\textup{s}}$ and an event happening at $\vect{x}_t$, the probability of the sensor detecting the event is expressed as
$$
\Pr(\text{Detected}|\vect{x}_{\textup{s}}^i,\vect{x}_t) = \beta^i\textup{e}^{-\alpha^i\frac{(\vect{x}_{\textup{s}}^i-\vect{x}_t)^\top(\vect{x}_{\textup{s}}^i-\vect{x}_t)}{{\gamma^i}^2}}, $$
where $\alpha^i,\beta^i,\gamma^i$ are sensor $i$'s parameters. In this case, the QoS of the sensor $i$ at location $\vect{x}_{\textup{s}}^i$ over the 2-D space $\vect{x}\in\real^2$ can be defined as $Q(\vect{x}|\vect{x}_{\textup{s}}^i)=\Pr(\text{Detected}|\vect{x}_{\textup{s}}^i,\vect{x})=z^i\mathcal{N}(\vect{x}|\vect{x}_{\textup{s}}^i,\vect{\Lambda}^i)$, where $z^i=\sqrt{2\pi|\vect{\Lambda}^i|}\beta^i$, $\vect{\Lambda}^i=\left[\begin{smallmatrix}
\sigma^i & 0\\
0 & \sigma^i
\end{smallmatrix}\right]$, $\sigma^i=\frac{{\gamma^i}^2}{2\alpha^i}$.
In this example, we consider a more general sensor model with anisotropic sensory capability, i.e. 
QoS is  $Q(\vect{x}|\vect{x}_{\textup{s}}^i,\theta_{\textup{s}}^i)=z^i\mathcal{N}(\vect{x}|\vect{x}_{\textup{s}}^i,\vect{\Sigma}^i(\theta_{\textup{s}}^i))$, 
with  $\vect{\Sigma}^i({\theta_{\textup{s}}^i})=\vect{R}(\theta_{\textup{s}}^i)\vect{\Lambda}^i\vect{R}^\top(\theta_{\textup{s}}^i)$, $\vect{\Lambda}^i=\left[\begin{smallmatrix}
\sigma_x^i & 0\\
0 & \sigma_y^i
\end{smallmatrix}\right]$, and $\theta_{\textup{s}}^i$ is the orientation of sensor $i$. 
Lastly, the collective density distribution of QoS provided by the 6 sensors is  $$q(\vect{x}|\{\vect{x}^i_{\textup{s}},\vect{\Sigma}^i(\theta_{\textup{s}}^i)\}_{i\in\VV_{\textup{s}}})=\sum_{i\in\VV_{\textup{s}}}\omega_{\textup{s}}^i \mathcal{N}(\vect{x}|\vect{x}_{\textup{s}}^i,\vect{\Sigma}^i(\theta_{\textup{s}}^i))$$ where $\omega^i_{\textup{s}}=\frac{z^i}{\sum_{i=1}^N z^i}$.
We set $\omega_{\textup{s}}^1=0.15,\sigma_x^1=70,\sigma_y^1=25$, $\omega_{\textup{s}}^2=0.15,\sigma_x^2=30,\sigma_y^2=15$, $\omega_{\textup{s}}^3=0.2,\sigma_x^3=80,\sigma_y^1=30$, $\omega_{\textup{s}}^4=0.1,\sigma_x^4=30,\sigma_y^4=30$, $\omega_{\textup{s}}^5=0.1,\sigma_x^5=60,\sigma_y^5=40$, and
$\omega_{\textup{s}}^6=0.3,\sigma_x^6=30,\sigma_y^6=30$.  See Fig.~\ref{fig::QoS_sensor} for a graphical demonstration of QoS of agent $1$. 

Each agent $i\in\VV_{\textup{s}}$ evaluates its costs $C_{ik}^\star$ for all $k\in\mathcal{K}$ by \eqref{eq::minimum_KLD}. Then, the agents cooperatively solve the distributed multi-agent assignment problem \eqref{eq::linear_programming_standard} by the means of distributed simplex algorithm \cite{BM-NG-BF-AF:12}. The optimal assignment plan of \eqref{eq::linear_programming_standard} is $Z_{14}^\star=1$, $Z_{22}^\star=1$, $Z_{33}^\star=1$, $Z_{41}^\star=1$, $Z_{56}^\star=1$, and $Z_{65}^\star=1$. 

Suppose the service agents $i\in\VV_{\textup{s}},$ motion is described by a unicycle dynamics
\begin{align}
    \dot{x}_{\textup{s},x}^i=v^i\cos{\theta}_{\textup{s}}^i,\quad
    \dot{x}_{\textup{s},y}^i=v^i\sin{\theta}_{\textup{s}}^i,\quad
    \dot{\theta}_{\textup{s}}^i&=\omega^i,
\end{align}
where $[x_{\textup{s},x}^i \quad x_{\textup{s},y}^i]^\top=\vect{x}_{\textup{s}}^i$, $v^i\in\real$ and $\omega^i\in\real$ are are linear velocity and angular velocity of each agent $i$, respectively. The agents execute the feedback linearization procedure \cite{GO-ADL-MV:02} to achieve the equivalent linear dynamics
\begin{align}
    \left[\begin{smallmatrix}\dot{\chi}^i_1\\ \dot{\chi}^i_2\\ \dot{\chi}^i_3\\ \dot{\chi}^i_4 \end{smallmatrix}\right]&=\underbrace{\left[\begin{smallmatrix}0&1 &0&0\\0&0&0&0\\0&0&0&1\\0&0&0&0 \end{smallmatrix}\right]}_{\vect{A}} \left[\begin{smallmatrix}\chi^i_1\\ \chi^i_2\\ \chi^i_3\\ \chi^i_4 \end{smallmatrix}\right]+\underbrace{\left[\begin{smallmatrix}0&0\\1&0\\0&0\\0&1\end{smallmatrix}\right]}_{\vect{B}}\left[\begin{smallmatrix}u^i_1\\u^i_2\end{smallmatrix}\right],
    \end{align}
    by the change of variables
    $\chi^i_1=x_{\textup{s},x}^i$, $\chi^i_2=v^i\cos{\theta}_{\textup{s}}^i$, $\chi^i_3=x_{\textup{s},y}^i$, $\chi^i_4=v^i\sin{\theta}_{\textup{s}}^i$
and the compensator
\begin{align}
\begin{split}
  \dot{v}^i&=u^i_1\cos{\theta}_{\textup{s}}^i +u^i_2\sin{\theta}_{\textup{s}}^i,\\
  \omega^i&=\frac{u^i_2\cos{\theta}_{\textup{s}}^i-u^i_1\sin{\theta}_{\textup{s}}^i}{v^i}.
\end{split}
\end{align}
Finally, control \eqref{eq::min_energy_control} is applied to to drive the agents to their desired state $\vect{\chi}^{i^\star}=[x_{\textup{s},x}^{i^\star} \quad v^{i^\star}\cos{\theta_{\textup{s}}^{i^\star}} \quad x_{\textup{s},y}^{i^\star} \quad v^{i^\star}\sin{\theta_{\textup{s}}^{i^\star}}]^\top$ corresponding to their assigned subregions and the optimal pose, i.e., $\vect{x}_{\textup{s}}^{i^\star}=\vect{\mu}_k^i$ and $\theta_{\textup{s}}^{i^\star}=\theta_k^i$ if $Z_{ik}^\star=1$, where $v^{i^\star}>0$ is the arrival velocity which we can assign. 

The density distribution of QoS provided by the $6$ agents (sensors) after deployment is illustrated in Fig.~\ref{fig::sensor_QoS_distribution}, where the black circles are the targets, the colored dots represent the 6 agents and the blue color map is the collective QoS distribution. The darker blue indicates the better QoS.  We can see that the collective QoS distribution is similar to the targets’ distribution. Hence, with the proposed two-stage deployment strategy, the distribution of QoS efficiently covers the targets.




\begin{figure}
  \centering
  \includegraphics[width=0.35\textwidth]{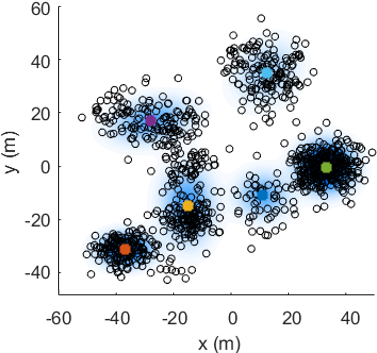}
  \caption{{\small The QoS distribution provided by the 6 sensors.}}
  \label{fig::sensor_QoS_distribution}
\end{figure}

\section{Conclusion}\label{sec::conclusion}
We proposed a two-stage distributed deployment solution for mobile agents to efficiently cover a group of dense targets with their service, i.e., the resulting collective QoS distribution is as much similar to as possible to the target density distribution.
In the first stage, we proposed the distributed consensus-based EM algorithm to enable all agents to estimate the target density distribution by a GMM.
In the second stage, we first assessed the smallest KLD of each pair of agent's Gaussian QoS distribution and Gaussian basis subregion. 
Then, we formulated a distributed multi-agent assignment problem to allocate each service agent to a subregion by taking the smallest KLD between the agent and the subregion as the assignment cost. As a result, the summation of the KLD of each assigned agent-subregion pair is minimal.
We illustrated the effectiveness of our proposed deployment solution using an application example for event detection via mobile sensors. 


\bibliographystyle{ieeetr}%
\bibliography{bib/alias,bib/Reference} 

\begin{thebibliography}{10}

\bibitem{KM-Ck-HB:11}
M.~Kumar, K.~Cohen, and B.~HomChaudhuri, ``Cooperative control of multiple
  uninhabited aerial vehicles for monitoring and fighting wildfires,'' {\em
  Journal of Aerospace Computing, Information, and Communication}, vol.~8,
  no.~1, pp.~1--16, 2011.

\bibitem{FA-MAN-MHJ-GMC-TAJ-KSB:14}
A.~Faustine, A.~N. Mvuma, H.~J. Mongi, M.~C. Gabriel, A.~J. Tenge, and S.~B.
  Kucel, ``Wireless sensor networks for water quality monitoring and control
  within lake victoria basin,'' 2014.

\bibitem{SS-BF-MS:08}
S.~Susca, F.~Bullo, and S.~Martinez, ``Monitoring environmental boundaries with
  a robotic sensor network,'' {\em IEEE Transactions on Control Systems
  Technology}, vol.~16, no.~2, pp.~288--296, 2008.

\bibitem{CJ-MS-KT-BF:04}
J.~Cortes, S.~Martinez, T.~Karatas, and F.~Bullo, ``Coverage control for mobile
  sensing networks,'' {\em IEEE Transactions on robotics and Automation},
  vol.~20, no.~2, pp.~243--255, 2004.

\bibitem{BF-CJ-MS:09}
F.~Bullo, J.~Cortes, and S.~Martinez, {\em Distributed control of robotic
  networks: a mathematical approach to motion coordination algorithms},
  vol.~27.
\newblock Princeton University Press, 2009.

\bibitem{WG-CG-LP-TF:06}
G.~Wang, G.~Cao, and T.~F. La~Porta, ``Movement-assisted sensor deployment,''
  {\em IEEE Transactions on Mobile Computing}, vol.~5, no.~6, pp.~640--652,
  2006.

\bibitem{PL-KV-MR-PG:08}
L.~C. Pimenta, V.~Kumar, R.~C. Mesquita, and G.~A. Pereira, ``Sensing and
  coverage for a network of heterogeneous robots,'' in {\em 2008 47th IEEE
  conference on decision and control}, pp.~3947--3952, IEEE, 2008.

\bibitem{PA-FLC-PL-SM:15}
A.~Pierson, L.~C. Figueiredo, L.~C. Pimenta, and M.~Schwager, ``Adapting to
  performance variations in multi-robot coverage,'' in {\em 2015 IEEE
  international conference on robotics and automation (ICRA)}, pp.~415--420,
  IEEE, 2015.

\bibitem{AO-KDE:16}
O.~Arslan and D.~E. Koditschek, ``Voronoi-based coverage control of
  heterogeneous disk-shaped robots,'' in {\em 2016 IEEE International
  Conference on Robotics and Automation (ICRA)}, pp.~4259--4266, IEEE, 2016.

\bibitem{LK-KJ:09}
K.~Laventall and J.~Cort{\'e}s, ``Coverage control by multi-robot networks with
  limited-range anisotropic sensory,'' {\em International Journal of Control},
  vol.~82, no.~6, pp.~1113--1121, 2009.

\bibitem{FF-ZX-CX-ZT:17}
F.~Farzadpour, X.~Zhang, X.~Chen, and T.~Zhang, ``On performance measurement
  for a heterogeneous planar field sensor network,'' in {\em 2017 IEEE
  International Conference on Advanced Intelligent Mechatronics (AIM)},
  pp.~166--171, IEEE, 2017.

\bibitem{GA-HS-HT-FM:08}
A.~Gusrialdi, S.~Hirche, T.~Hatanaka, and M.~Fujita, ``Voronoi based coverage
  control with anisotropic sensors,'' in {\em 2008 American control
  conference}, pp.~736--741, IEEE, 2008.

\bibitem{SM-RD-SJ:09}
M.~Schwager, D.~Rus, and J.-J. Slotine, ``Decentralized, adaptive coverage
  control for networked robots,'' {\em The International Journal of Robotics
  Research}, vol.~28, no.~3, pp.~357--375, 2009.

\bibitem{CA-TM-CR-SL-PG:15}
A.~Carron, M.~Todescato, R.~Carli, L.~Schenato, and G.~Pillonetto,
  ``Multi-agents adaptive estimation and coverage control using gaussian
  regression,'' in {\em 2015 European Control Conference (ECC)},
  pp.~2490--2495, IEEE, 2015.

\bibitem{SF-MM-ZY-GBW:15}
F.~Sharifi, M.~Mirzaei, Y.~Zhang, and B.~W. Gordon, ``Cooperative multi-vehicle
  search and coverage problem in an uncertain environment,'' {\em Unmanned
  systems}, vol.~3, no.~01, pp.~35--47, 2015.

\bibitem{BM-NG-BF-AF:12}
M.~B{\"u}rger, G.~Notarstefano, F.~Bullo, and F.~Allg{\"o}wer, ``A distributed
  simplex algorithm for degenerate linear programs and multi-agent
  assignments,'' {\em Automatica}, vol.~48, no.~9, pp.~2298--2304, 2012.

\bibitem{MDJC:03}
D.~J. MacKay, {\em Information theory, inference and learning algorithms}.
\newblock Cambridge university press, 2003.

\bibitem{HJR-ORA:07}
J.~R. Hershey and P.~A. Olsen, ``Approximating the kullback leibler divergence
  between gaussian mixture models,'' in {\em 2007 IEEE International Conference
  on Acoustics, Speech and Signal Processing-ICASSP'07}, vol.~4, pp.~IV--317,
  IEEE, 2007.

\bibitem{FB-JC-SM:09}
F.~Bullo, J.~Cort\'es, and S.~Mart{\'\i}nez, {\em Distributed Control of
  Robotic Networks}.
\newblock Applied Mathematics Series, Princeton University Press, 2009.

\bibitem{YC-SS:20}
Y.-F. Chung and S.~Kia, ``Dynamic active average consensus and its application
  in containment control,'' {\em arXiv:2008.05722}, 2020.

\bibitem{HJR-OPA:07}
J.~R. Hershey and P.~A. Olsen, ``Approximating the kullback leibler divergence
  between gaussian mixture models,'' in {\em 2007 IEEE International Conference
  on Acoustics, Speech and Signal Processing-ICASSP'07}, vol.~4, pp.~IV--317,
  IEEE, 2007.

\bibitem{DAP-LNM-RDB:77}
A.~P. Dempster, N.~M. Laird, and D.~B. Rubin, ``Maximum likelihood from
  incomplete data via the em algorithm,'' {\em Journal of the Royal Statistical
  Society: Series B}, vol.~39, no.~1, pp.~1--22, 1977.

\bibitem{PG-CM:19}
G.~Peyr{\'e} and M.~Cuturi, ``Computational optimal transport,'' {\em
  Foundations and Trends in Machine Learning}, vol.~11, no.~5-6, pp.~355--607,
  2019.

\bibitem{BDP:98}
D.~P. Bertsekas, {\em Network optimization: continuous and discrete models}.
\newblock Athena Scientific Belmont, MA, 1998.

\bibitem{FL-DV-VS:12}
F.~Lewis, D.~Vrabie, and V.~Syrmos, {\em Optimal Control}.
\newblock Wiley, 2012.

\bibitem{BF-FT-DWHF:03}
F.~Bourgault, T.~Furukawa, and H.~F. Durrant-Whyte, ``Optimal search for a lost
  target in a bayesian world,'' in {\em Field and service robotics},
  pp.~209--222, Springer, 2003.

\bibitem{XH-CR-XD:17}
H.~Xiao, R.~Cui, and D.~Xu, ``A sampling-based bayesian approach for
  cooperative multiagent online search with resource constraints,'' {\em IEEE
  Transactions on Cybernetics}, vol.~48, no.~6, pp.~1773--1785, 2017.

\bibitem{GO-ADL-MV:02}
G.~Oriolo, A.~D. Luca, and M.~Vendittelli, ``\uppercase{WMR} control via
  dynamic feedback linearization: design, implementation, and experimental
  validation,'' {\em Transactions on Control Systems Technology}, vol.~10,
  pp.~835--852, 2002.

\end{thebibliography}

\end{document}